\documentclass[12pt]{article}
\usepackage{amsthm,amsmath,amssymb}
\usepackage{multirow}

\theoremstyle{plain}
\newtheorem{theorem}{Theorem}
\newtheorem{lemma}[theorem]{Lemma}

\title{A New Lower Bound for the Ramsey Number $R(4, 8)$}

\author{Hiroshi Fujita\\
\small Department of Informatics\\[-0.8ex]
\small Kyushu University\\[-0.8ex] 
\small 744 Motooka, Nishi-ku, Fukuoka 819-0395 Japan\\
\small\tt fujita@inf.kyushu-u.ac.jp}

\date{}

\begin{document}

\maketitle

\begin{abstract}
The lower bound for the classical Ramsey number $R(4, 8)$ is improved
from 56 to 58.
The author has found a new edge coloring of $K_{57}$
that has no complete graphs of order 4 in the first color,
and no complete graphs of order 8 in the second color.
The coloring was found using a SAT solver
which is based on MiniSat and customized for solving Ramsey problems.
\end{abstract}

Recently Exoo improved the lower bound for the classical Ramsey number
$R(4, 6)$ \cite{R_4_6}.
This note deals with a new lower bound for $R(4, 8)$.
The classical Ramsey number $R(s, t)$ is the smallest integer $n$
such that in any two-coloring of the edges of $K_n$
there is a monochromatic copy of $K_s$ in the first color or
a monochromatic copy of $K_t$ in the second color.
Some of the interesting instances can be found at Exoo's web site \cite{Exoo}
and at McKay's web site \cite{McKay}.
A recent summary of the state of the art for Ramsey numbers
can be found in the Dynamic Survey \cite{DS}.

Exoo writes that some unsettled cases for two color classical Ramsey numbers
such as $R(4, 6), R(3, 10),$ and $R(5, 5)$
can only be solved by using computer methods.
The author think SAT solvers can be one of the promising tools
to do this kind of work.
Here we try to obtain a Ramsey graph $R(s, t, n)$ \cite{McKay}
by encoding the condition for $R(s, t, n)$ to exist
into a conjunctive normal form (CNF), called Ramsey clauses, as follows:
\[ C_{(s, t, n)}:
\Big( \bigwedge_{K_{s}\subset K_{n}}
 \vee_{e_{ij}\in K_{s}}
 ~\neg e_{ij} \Big)
\wedge
\Big( \bigwedge_{K_{t}\subset K_{n}}
 \vee_{e_{ij}\in K_{t}}
 ~e_{ij} \Big)
\]
\noindent
where each $e_{ij}$ is the propositional variable, called Ramsey variables,
for the edge between vertices $i$ and $j$,
and assigning $true$ to it means the edge is colored in the first color,
otherwise the second.
If $C_{(s, t, n)}$ is unsatisfiable, then no $R(s, t, n)$ exists.
Otherwise a model representing a $R(s, t, n)$ can be obtained.
To determine whether $C_{(s, t, n)}$ is satisfiable or not,
and to obtain a model when it is satisfiable,
one can use a state of the art SAT solver such as MiniSat \cite{MiniSat}.

However, it is extremely difficult to achieve our task
using such a na\"ive setting of the problem as above,
since the search space becomes enormous.
In fact, the number of variables is $\binom{n}{2}$ and
the number of clauses is $\binom{n}{s}+\binom{n}{t}$,
thus for instance,
903 variables and 1925196 clauses for $C_{(5,5,43)}$
which is fairly large, though still manageable.
Those for other interesting cases, say $C_{(3,10,40)}$,
are beyond our reach.

So we consider imposing some constraints on the problem
in order for it to be simplified,
and the search space of which being reduced.
One of the most straightforward yet significantly effective constraints is:
\[ e_{ij} \equiv z_{k} \quad (0\le i<j<n,~ j-i=k) \]
\noindent
where $z_{k}~(1\le k\le n-1)$ are new propositional variables,
called Z-variables.
The constraint, called Z-constraint, is represented in CNF,
called Z-clauses, as follows:
\[ Z_{n}:
  \bigwedge_{0\le i<j<n}
\big(
  (\neg e_{ij} \vee z_{j-i})
\wedge
  (e_{ij} \vee \neg z_{j-i})
\big)
\]

\begin{figure}[!ht]
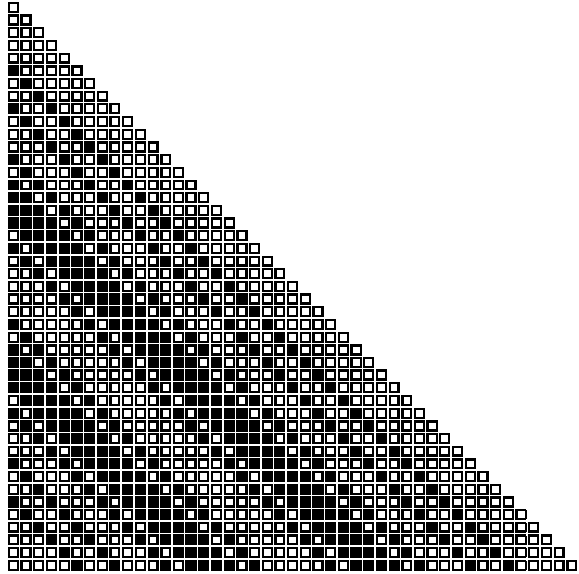

\begin{center}
\begin{minipage}{0.6\textwidth}
\setlength{\parindent}{0pt}
\setlength{\baselineskip}{0pt}
\newcommand*{\x}{{\fboxsep 0pt\fbox{\rule{3pt}{3pt}}\hspace{1pt}}}
\newcommand*{\q}{{\fboxsep 0pt\fbox{\rule{0pt}{3pt}\rule{3pt}{0pt}}\hspace{1pt}}}
\newcommand*{\mynl}{\hfil\newline}
\q\mynl
\q\q\mynl
\q\q\q\mynl
\q\q\q\q\mynl
\q\q\q\q\q\mynl
\x\q\q\q\q\q\mynl
\q\x\q\q\q\q\q\mynl
\q\q\x\q\q\q\q\q\mynl
\x\q\q\x\q\q\q\q\q\mynl
\q\x\q\q\x\q\q\q\q\q\mynl
\q\q\x\q\q\x\q\q\q\q\q\mynl
\q\q\q\x\q\q\x\q\q\q\q\q\mynl
\x\q\q\q\x\q\q\x\q\q\q\q\q\mynl
\q\x\q\q\q\x\q\q\x\q\q\q\q\q\mynl
\x\q\x\q\q\q\x\q\q\x\q\q\q\q\q\mynl
\x\x\q\x\q\q\q\x\q\q\x\q\q\q\q\q\mynl
\x\x\x\q\x\q\q\q\x\q\q\x\q\q\q\q\q\mynl
\x\x\x\x\q\x\q\q\q\x\q\q\x\q\q\q\q\q\mynl
\q\x\x\x\x\q\x\q\q\q\x\q\q\x\q\q\q\q\q\mynl
\x\q\x\x\x\x\q\x\q\q\q\x\q\q\x\q\q\q\q\q\mynl
\q\x\q\x\x\x\x\q\x\q\q\q\x\q\q\x\q\q\q\q\q\mynl
\q\q\x\q\x\x\x\x\q\x\q\q\q\x\q\q\x\q\q\q\q\q\mynl
\q\q\q\x\q\x\x\x\x\q\x\q\q\q\x\q\q\x\q\q\q\q\q\mynl
\q\q\q\q\x\q\x\x\x\x\q\x\q\q\q\x\q\q\x\q\q\q\q\q\mynl
\q\q\q\q\q\x\q\x\x\x\x\q\x\q\q\q\x\q\q\x\q\q\q\q\q\mynl
\x\q\q\q\q\q\x\q\x\x\x\x\q\x\q\q\q\x\q\q\x\q\q\q\q\q\mynl
\q\x\q\q\q\q\q\x\q\x\x\x\x\q\x\q\q\q\x\q\q\x\q\q\q\q\q\mynl
\x\q\x\q\q\q\q\q\x\q\x\x\x\x\q\x\q\q\q\x\q\q\x\q\q\q\q\q\mynl
\x\x\q\x\q\q\q\q\q\x\q\x\x\x\x\q\x\q\q\q\x\q\q\x\q\q\q\q\q\mynl
\x\x\x\q\x\q\q\q\q\q\x\q\x\x\x\x\q\x\q\q\q\x\q\q\x\q\q\q\q\q\mynl
\x\x\x\x\q\x\q\q\q\q\q\x\q\x\x\x\x\q\x\q\q\q\x\q\q\x\q\q\q\q\q\mynl
\q\x\x\x\x\q\x\q\q\q\q\q\x\q\x\x\x\x\q\x\q\q\q\x\q\q\x\q\q\q\q\q\mynl
\x\q\x\x\x\x\q\x\q\q\q\q\q\x\q\x\x\x\x\q\x\q\q\q\x\q\q\x\q\q\q\q\q\mynl
\q\x\q\x\x\x\x\q\x\q\q\q\q\q\x\q\x\x\x\x\q\x\q\q\q\x\q\q\x\q\q\q\q\q\mynl
\q\q\x\q\x\x\x\x\q\x\q\q\q\q\q\x\q\x\x\x\x\q\x\q\q\q\x\q\q\x\q\q\q\q\q\mynl
\q\q\q\x\q\x\x\x\x\q\x\q\q\q\q\q\x\q\x\x\x\x\q\x\q\q\q\x\q\q\x\q\q\q\q\q\mynl
\x\q\q\q\x\q\x\x\x\x\q\x\q\q\q\q\q\x\q\x\x\x\x\q\x\q\q\q\x\q\q\x\q\q\q\q\q\mynl
\q\x\q\q\q\x\q\x\x\x\x\q\x\q\q\q\q\q\x\q\x\x\x\x\q\x\q\q\q\x\q\q\x\q\q\q\q\q\mynl
\q\q\x\q\q\q\x\q\x\x\x\x\q\x\q\q\q\q\q\x\q\x\x\x\x\q\x\q\q\q\x\q\q\x\q\q\q\q\q\mynl
\x\q\q\x\q\q\q\x\q\x\x\x\x\q\x\q\q\q\q\q\x\q\x\x\x\x\q\x\q\q\q\x\q\q\x\q\q\q\q\q\mynl
\q\x\q\q\x\q\q\q\x\q\x\x\x\x\q\x\q\q\q\q\q\x\q\x\x\x\x\q\x\q\q\q\x\q\q\x\q\q\q\q\q\mynl
\q\q\x\q\q\x\q\q\q\x\q\x\x\x\x\q\x\q\q\q\q\q\x\q\x\x\x\x\q\x\q\q\q\x\q\q\x\q\q\q\q\q\mynl
\q\q\q\x\q\q\x\q\q\q\x\q\x\x\x\x\q\x\q\q\q\q\q\x\q\x\x\x\x\q\x\q\q\q\x\q\q\x\q\q\q\q\q\mynl
\q\q\q\q\x\q\q\x\q\q\q\x\q\x\x\x\x\q\x\q\q\q\q\q\x\q\x\x\x\x\q\x\q\q\q\x\q\q\x\q\q\q\q\q\mynl
\q\q\q\q\q\x\q\q\x\q\q\q\x\q\x\x\x\x\q\x\q\q\q\q\q\x\q\x\x\x\x\q\x\q\q\q\x\q\q\x\q\q\q\q\q\mynl
\end{minipage}
\end{center}
\caption{\label{fig:R_4_7_46}A $(4, 7)$-coloring of $K_{46}$
satisfying the full $Z_{46}$.}
\end{figure}

We try to solve $C_{(s, t, n)} \wedge Z_{n}$.
If it is satisfiable then $C_{(s, t, n)}$ is also satisfiable,
and the obtained model should represent a Ramsey graph $R(s, t, n)$
(called Z-Ramsey graph denoted by $R^{Z}(s, t, n)$ in the sequel).
Now, the number of variables that need to be decided is only $n-1$,
and the problem becomes drastically easier.

In this way a Ramsey graph $R^{Z}(4, 7, 46)$, for instance, was obtained
as shown in Figure \ref{fig:R_4_7_46} very easily.
The figure depicts the lower triangle of the adjacency matrix
excluding the main diagonal.
A black filled box means $true$ is assigned to a variable
$e_{ij}$, i.e. the edge between vertices $i$ and $j$
is colored in the first color,
and a white box $false$, i.e. the second color.

As a matter of fact,
instead of using a standard SAT solver,
the author used another program specially designed to obtain
a $R^{Z}(s, t, n)$,
which does not require memory space to store $C_{(s,t,n)}$ clauses for it,
but computes all of the required conditions on the fly,
thus being able to deal with problems for larger $n$,
even more than 100, within reasonable CPU time.

\begin{table}[!b]
\caption{\label{tbl:zramsey}The largest Z-Ramsey graphs $R^Z(s,t,n)$}
\begin{center}
\begin{tabular}{|l||c|c|c|c|c|c|c|c|}
\hline
$s\,\backslash\,t$ & 3 & 4 & 5 & 6 & 7 & 8 & 9 & 10 \\ \hline\hline
\multirow{2}{*}{3}
  & 5 & 8 & 13 & 16 & 21 & 26 & 35 & 38 \\
  & 2 & 2 & 3 & 7 & 13 & 13 & 4 & 21 \\ \hline
\multirow{2}{*}{4}
  &  & 17 & 24 & 33 & 46 & 52 & 68 & 91 \\
  &  & 2 & 6 & 24 & 21 & 13 & 487 & 35 \\ \hline
\multirow{2}{*}{5}
  &  &  & 41 & 56 & 79 & \multirow{2}{*}{--} & \multirow{2}{*}{--} & \multirow{2}{*}{--} \\
  &  &  & 22 & 12 & 49 &  &  &  \\ \hline
\multirow{2}{*}{6}
  &  &  &  & 101 & \multirow{2}{*}{--} & \multirow{2}{*}{--} & \multirow{2}{*}{--} & \multirow{2}{*}{--} \\
  &  &  &  & 2 &  &  &  &  \\ \hline
\multicolumn{9}{l}{top: the largest $n$, bottom: \# of graphs}
\end{tabular}
\end{center}
\end{table}

Table \ref{tbl:zramsey} shows
the largest $n$'s for Z-Ramsey graphs $R^Z(s,t,n)$ and their numbers
obtained by the special program.
Consider a special Z-constraint denoted by Z$^s$
by adding further constraint as follows:
\[ z_{n-k} \equiv z_{k} \quad (1\le k<n) \]
Then, for all entries in the table, except for $R^Z(4,8,52)$,
there exist the largest Z-Ramsey graphs satisfying Z$^s$.

It turns out that there is no model for $C_{(4,7,n)} \wedge Z_{n},~n>46$.
Still the author expect that there exists one
which satisfies large parts of $Z_{n}$ and is highly symmetric.
For this, we somehow relax the constraint $Z_{n}$ in the following way.
\begin{itemize}
\item
Each time a Z-clause causes a conflict during a search,
increase its penalty score,
and continue the search disregarding the conflict.
\item
When a search fails after all,
pick up some number of Z-clauses having higher penalty score,
remove them from the current set of Z-clauses,
then restart a new search with the relaxed Z-constraint.
\item
Repeat the above relaxation-and-restarts until
a solution is obtained or the attempt totally fails
with whole Z-clauses being removed.
\end{itemize}
Thus the author obtained a $R(4,7,47)$ and a $R(4,7,48)$.


If $Z_{n}$ itself seems too strong, one may relax it
manually in several ways as follows:
\begin{itemize}
\item {\bf imperfect Z.}~~
Omitting $Z_n$ clauses for some $n$'s.
For instance, a Ramsey graph $R(4,7,48)$ can be obtained
rather easily by omitting only such constraints that refer to $Z_{10}$.
\item {\bf partitioned Z.}~~
Using $\{Z_{n}\}_{p}$, several sets of $Z_{n}$ variables,
where $p$ is a function of $(row,column)$ of the adjacency matrix
for the graph under consideration.
For instance, a Ramsey graph $R(4,7,48)$ would be obtained
by using the following $p$.
\[
p(row, column) =
\begin{cases}
 0 & ( 24\le row\le 33  ) \\
 1 & ( otherwise )
\end{cases}
\]
\end{itemize}

\begin{lemma}
  \label{Lem:R_4_7_48}
  $R(4, 7)\ge 49$
\end{lemma}

\begin{proof}
The proof is given by the coloring of $K_{48}$
which can be derived from Figure \ref{fig:R_4_7_48}.
\end{proof}

\begin{figure}[!ht]
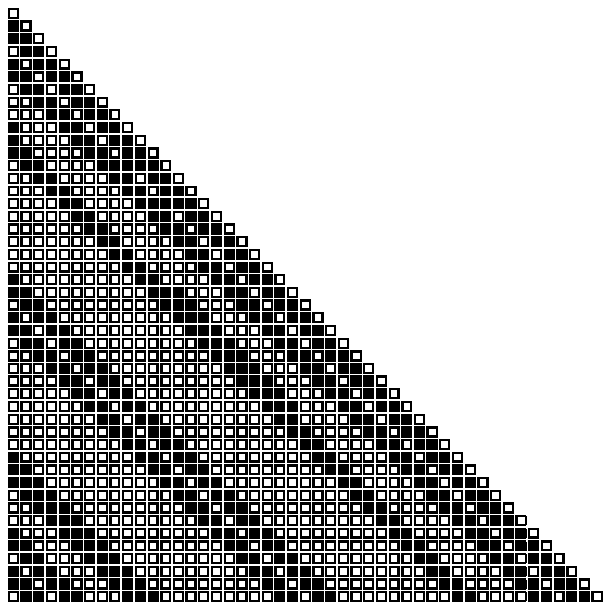

\begin{center}
\begin{minipage}{0.6\textwidth}
\setlength{\parindent}{0pt}
\setlength{\baselineskip}{0pt}
\newcommand*{\x}{{\fboxsep 0pt\fbox{\rule{3pt}{3pt}}\hspace{1pt}}}
\newcommand*{\q}{{\fboxsep 0pt\fbox{\rule{0pt}{3pt}\rule{3pt}{0pt}}\hspace{1pt}}}
\newcommand*{\mynl}{\hfil\newline}
\q\mynl
\x\q\mynl
\x\x\q\mynl
\q\x\x\q\mynl
\x\q\x\x\q\mynl
\x\x\q\x\x\q\mynl
\q\x\x\q\x\x\q\mynl
\q\q\x\x\q\x\x\q\mynl
\q\q\q\x\x\q\x\x\q\mynl
\x\q\q\q\x\x\q\x\x\q\mynl
\x\q\q\q\q\x\x\q\x\x\q\mynl
\x\x\q\q\q\q\x\x\q\x\x\q\mynl
\q\x\x\q\q\q\q\x\x\x\x\x\q\mynl
\q\q\x\x\q\q\q\q\x\x\q\x\x\q\mynl
\q\q\q\x\x\q\q\q\q\x\x\q\x\x\q\mynl
\q\q\q\q\x\x\q\q\q\q\x\x\x\x\x\q\mynl
\q\q\q\q\q\x\x\q\q\q\q\x\x\q\x\x\q\mynl
\q\q\q\q\q\q\x\x\q\q\q\q\x\x\q\x\x\q\mynl
\q\q\q\q\q\q\q\x\x\q\q\q\q\x\x\q\x\x\q\mynl
\q\q\q\q\q\q\q\q\x\x\q\q\q\q\x\x\q\x\x\q\mynl
\q\q\q\q\q\q\q\q\q\x\x\q\q\q\q\x\x\q\x\x\q\mynl
\x\q\q\q\q\q\q\q\q\q\x\x\q\q\q\q\x\x\q\x\x\q\mynl
\x\x\q\q\q\q\q\q\q\q\q\x\x\x\q\q\q\x\x\q\x\x\q\mynl
\q\x\x\q\q\q\q\q\q\q\q\q\x\x\x\q\q\q\x\x\q\x\x\q\mynl
\x\q\x\x\q\q\q\q\q\q\q\q\q\x\x\x\q\q\q\x\x\q\x\x\q\mynl
\x\x\q\x\x\q\q\q\q\q\q\q\q\q\x\x\x\q\q\q\x\x\q\x\x\q\mynl
\q\x\x\q\x\x\q\q\q\q\q\q\q\q\q\x\x\x\q\q\q\x\x\q\x\x\q\mynl
\q\q\x\x\q\x\x\q\q\q\q\q\q\q\q\q\x\x\x\q\q\q\x\x\q\x\x\q\mynl
\q\q\q\x\x\q\x\x\q\q\q\q\q\q\q\q\q\x\x\x\q\q\q\x\x\q\x\x\q\mynl
\q\q\q\q\x\x\q\x\x\q\q\q\q\q\q\q\q\q\x\x\x\q\q\q\x\x\q\x\x\q\mynl
\q\q\q\q\q\x\x\q\x\x\q\q\q\q\q\q\q\q\q\x\x\x\q\q\q\x\x\q\x\x\q\mynl
\q\q\q\q\q\q\x\x\q\x\x\q\q\q\q\q\q\q\q\q\x\x\x\q\q\q\x\x\q\x\x\q\mynl
\q\q\q\q\q\q\q\x\x\q\x\x\q\q\q\q\q\q\q\q\q\x\x\q\q\q\q\x\x\q\x\x\q\mynl
\q\q\q\q\q\q\q\q\x\x\q\x\x\q\q\q\q\q\q\q\q\q\x\x\q\q\q\q\x\x\q\x\x\q\mynl
\q\q\q\q\q\q\q\q\q\x\x\q\x\x\q\q\q\q\q\q\q\q\q\x\x\q\q\q\q\x\x\q\x\x\q\mynl
\x\q\q\q\q\q\q\q\q\q\x\x\q\x\x\q\q\q\q\q\q\q\q\q\x\x\q\q\q\q\x\x\q\x\x\q\mynl
\x\x\q\q\q\q\q\q\q\q\q\x\x\q\x\x\q\q\q\q\q\q\q\q\q\x\x\q\q\q\q\x\x\q\x\x\q\mynl
\x\x\x\q\q\q\q\q\q\q\q\q\x\x\q\x\x\q\q\q\q\q\q\q\q\q\x\x\q\q\q\q\x\x\q\x\x\q\mynl
\q\x\x\x\q\q\q\q\q\q\q\q\q\x\x\q\x\x\q\q\q\q\q\q\q\q\q\x\x\q\q\q\q\x\x\q\x\x\q\mynl
\q\q\x\x\x\q\q\q\q\q\q\q\q\q\x\x\q\x\x\q\q\q\q\q\q\q\q\q\x\x\q\q\q\q\x\x\q\x\x\q\mynl
\q\q\q\x\x\x\q\q\q\q\q\q\q\q\q\x\x\q\x\x\q\q\q\q\q\q\q\q\q\x\x\q\q\q\q\x\x\q\x\x\q\mynl
\x\q\q\q\x\x\x\q\q\q\q\q\q\q\q\q\x\x\q\x\x\q\q\q\q\q\q\q\q\q\x\x\q\q\q\q\x\x\q\x\x\q\mynl
\x\x\q\q\q\x\x\x\q\q\q\q\q\q\q\q\q\x\x\q\x\x\q\q\q\q\q\q\q\q\q\x\x\q\q\q\q\x\x\q\x\x\q\mynl
\q\x\x\q\q\q\x\x\x\q\q\q\q\q\q\q\q\q\x\x\q\x\x\q\q\q\q\q\q\q\q\q\x\x\q\q\q\q\x\x\q\x\x\q\mynl
\x\q\x\x\q\q\q\x\x\x\q\q\q\q\q\q\q\q\q\x\x\q\x\x\q\q\q\q\q\q\q\q\q\x\x\q\q\q\q\x\x\q\x\x\q\mynl
\x\x\q\x\x\q\q\q\x\x\x\q\q\q\q\q\q\q\q\q\x\x\q\x\x\q\q\q\q\q\q\q\q\q\x\x\q\q\q\q\x\x\q\x\x\q\mynl
\q\x\x\q\x\x\q\q\q\x\x\x\q\q\q\q\q\q\q\q\q\x\x\q\x\x\q\q\q\q\q\q\q\q\q\x\x\q\q\q\q\x\x\q\x\x\q\mynl
\end{minipage}
\end{center}
\caption{\label{fig:R_4_7_48}A $(4, 7)$-coloring of $K_{48}$
satisfying a relaxed $Z_{48}$.}
\end{figure}

The search started with 1128 Ramsey variables + 47 Z-variables
and 73823652 Ramsey clauses + 2256 Z-clauses.
It took 47007 seconds (about 9 hours)
on x86\_64 GNU/Linux,
Intel$^{\mbox{\textregistered}}$ Xeon$^{\mbox{\textregistered}}$ CPU X5680 @ 3.33GHz
with 100 GB memory.
The relaxations of Z-clauses were invoked once, and
about 90 \% of the initial Z-clauses being kept and satisfied at last.

In some cases, a Ramsey graph $R(s',t',n')$ can be obtained easily
using a smaller $R(s, t, n)$,
such that $s\le s',~t\le t',~n<n'$, typically $s=s',~t=t'-1$.
For instance,
\begin{itemize}
\item
a $R(3,10,39)$ using a $R^Z(3,9,35)$,
\item
$\{R(3,12,49), R(3,12,50), R(3,12,51)\}$ using a $R^Z(3,11,45)$,
\item
$\{R(4,8,53), R(4,8,54), R(4,8,55)\}$ using a $R^Z(4,7,46)$,
\item
$\{R(3,20,106), R(3,20,107), R(3,20,108)\}$ using a $R^Z(3,19,105)$,
\item
a $R(4,11,97)$ using a $R^Z(4,10,91)$,
\item
$\{R(6,7,107), R(6,7,108), R(6,7,109)\}$ using a $R^Z(6,6,101)$.
\end{itemize}

Though the author has been searching for some time
a $R(4,8,56)$ using a $R^Z(4,7,46)$,
the attempt has not yet been accomplished at the time of this writing.
Instead, the author has succeeded only recently to obtain
a $R(4,8,56)$ using $R(4,7,48)$
as shown in Figure \ref{fig:R_4_7_48},
and later a $R(4,8,57)$ using another $R(4,7,48)$.

In this setting,
part of edges $e_{ij}~(0\le i<j<48)$ are colored exactly
in the same way as in the coloring for the given $R(4, 7, 48)$.
The rest of edges $e_{ij}~(0\le i<57,~48\le j<56)$ are left undecided.
Some of the Ramsey clauses will be satisfied due to the partial colorings,
and only the unsettled Ramsey clauses are given to our SAT solver.
Moreover, we impose a Z-constraint on the unsettled edges,
and perform repeated relaxations and restarts of solving
in the manner described above.

\begin{figure}[t]
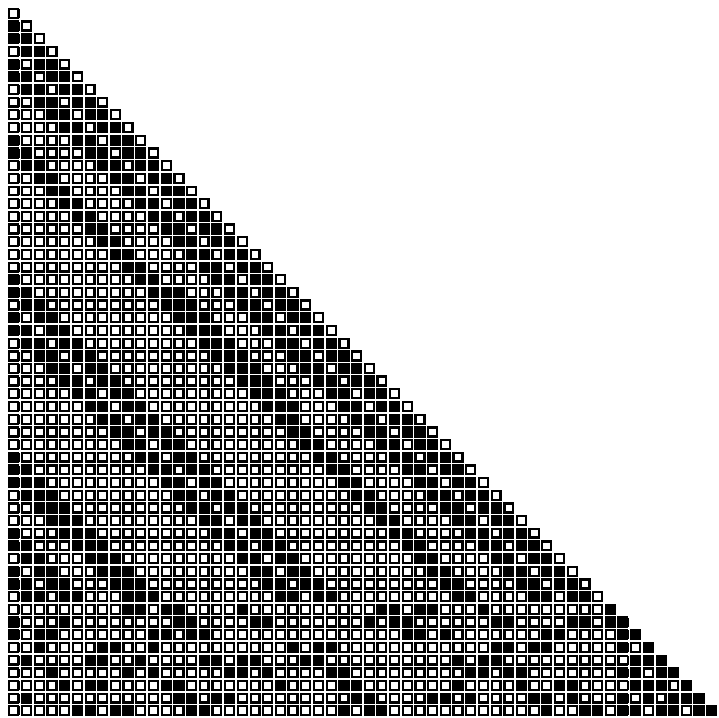

\begin{center}
\begin{minipage}{0.7\textwidth}
\setlength{\parindent}{0pt}
\setlength{\baselineskip}{0pt}
\newcommand*{\x}{{\fboxsep 0pt\fbox{\rule{3pt}{3pt}}\hspace{1pt}}}
\newcommand*{\q}{{\fboxsep 0pt\fbox{\rule{0pt}{3pt}\rule{3pt}{0pt}}\hspace{1pt}}}
\newcommand*{\mynl}{\hfil\newline}
\q\mynl
\x\q\mynl
\x\x\q\mynl
\q\x\x\q\mynl
\x\q\x\x\q\mynl
\x\x\q\x\x\q\mynl
\q\x\x\q\x\x\q\mynl
\q\q\x\x\q\x\x\q\mynl
\q\q\q\x\x\q\x\x\q\mynl
\q\q\q\q\x\x\q\x\x\q\mynl
\x\q\q\q\q\x\x\q\x\x\q\mynl
\x\x\q\q\q\q\x\x\q\x\x\q\mynl
\q\x\x\q\q\q\q\x\x\q\x\x\q\mynl
\q\q\x\x\q\q\q\q\x\x\q\x\x\q\mynl
\q\q\q\x\x\q\q\q\q\x\x\q\x\x\q\mynl
\q\q\q\q\x\x\q\q\q\q\x\x\q\x\x\q\mynl
\q\q\q\q\q\x\x\q\q\q\q\x\x\q\x\x\q\mynl
\q\q\q\q\q\q\x\x\q\q\q\q\x\x\q\x\x\q\mynl
\q\q\q\q\q\q\q\x\x\q\q\q\q\x\x\q\x\x\q\mynl
\q\q\q\q\q\q\q\q\x\x\q\q\q\q\x\x\q\x\x\q\mynl
\q\q\q\q\q\q\q\q\q\x\x\q\q\q\q\x\x\q\x\x\q\mynl
\x\q\q\q\q\q\q\q\q\q\x\x\q\q\q\q\x\x\q\x\x\q\mynl
\x\x\q\q\q\q\q\q\q\q\q\x\x\x\q\q\q\x\x\q\x\x\q\mynl
\q\x\x\q\q\q\q\q\q\q\q\q\x\x\x\q\q\q\x\x\q\x\x\q\mynl
\x\q\x\x\q\q\q\q\q\q\q\q\q\x\x\x\q\q\q\x\x\q\x\x\q\mynl
\x\x\q\x\x\q\q\q\q\q\q\q\q\q\x\x\x\q\q\q\x\x\q\x\x\q\mynl
\q\x\x\q\x\x\q\q\q\q\q\q\q\q\q\x\x\x\q\q\q\x\x\q\x\x\q\mynl
\q\q\x\x\q\x\x\q\q\q\q\q\q\q\q\q\x\x\x\q\q\q\x\x\q\x\x\q\mynl
\q\q\q\x\x\q\x\x\q\q\q\q\q\q\q\q\q\x\x\x\q\q\q\x\x\q\x\x\q\mynl
\q\q\q\q\x\x\q\x\x\q\q\q\q\q\q\q\q\q\x\x\x\q\q\q\x\x\q\x\x\q\mynl
\q\q\q\q\q\x\x\q\x\x\q\q\q\q\q\q\q\q\q\x\x\x\q\q\q\x\x\q\x\x\q\mynl
\q\q\q\q\q\q\x\x\q\x\x\q\q\q\q\q\q\q\q\q\x\x\x\q\q\q\x\x\q\x\x\q\mynl
\q\q\q\q\q\q\q\x\x\q\x\x\q\q\q\q\q\q\q\q\q\x\x\q\q\q\q\x\x\q\x\x\q\mynl
\q\q\q\q\q\q\q\q\x\x\q\x\x\q\q\q\q\q\q\q\q\q\x\x\q\q\q\q\x\x\q\x\x\q\mynl
\q\q\q\q\q\q\q\q\q\x\x\q\x\x\q\q\q\q\q\q\q\q\q\x\x\q\q\q\q\x\x\q\x\x\q\mynl
\x\q\q\q\q\q\q\q\q\q\x\x\q\x\x\q\q\q\q\q\q\q\q\q\x\x\q\q\q\q\x\x\q\x\x\q\mynl
\x\x\q\q\q\q\q\q\q\q\q\x\x\q\x\x\q\q\q\q\q\q\q\q\q\x\x\q\q\q\q\x\x\q\x\x\q\mynl
\x\x\x\q\q\q\q\q\q\q\q\q\x\x\q\x\x\q\q\q\q\q\q\q\q\q\x\x\q\q\q\q\x\x\q\x\x\q\mynl
\q\x\x\x\q\q\q\q\q\q\q\q\q\x\x\q\x\x\q\q\q\q\q\q\q\q\q\x\x\q\q\q\q\x\x\q\x\x\q\mynl
\q\q\x\x\x\q\q\q\q\q\q\q\q\q\x\x\q\x\x\q\q\q\q\q\q\q\q\q\x\x\q\q\q\q\x\x\q\x\x\q\mynl
\q\q\q\x\x\x\q\q\q\q\q\q\q\q\q\x\x\q\x\x\q\q\q\q\q\q\q\q\q\x\x\q\q\q\q\x\x\q\x\x\q\mynl
\x\q\q\q\x\x\x\q\q\q\q\q\q\q\q\q\x\x\q\x\x\q\q\q\q\q\q\q\q\q\x\x\q\q\q\q\x\x\q\x\x\q\mynl
\x\x\q\q\q\x\x\x\q\q\q\q\q\q\q\q\q\x\x\q\x\x\q\q\q\q\q\q\q\q\q\x\x\q\q\q\q\x\x\q\x\x\q\mynl
\q\x\x\q\q\q\x\x\x\q\q\q\q\q\q\q\q\q\x\x\q\x\x\q\q\q\q\q\q\q\q\q\x\x\q\q\q\q\x\x\q\x\x\q\mynl
\x\q\x\x\q\q\q\x\x\x\q\q\q\q\q\q\q\q\q\x\x\q\x\x\q\q\q\q\q\q\q\q\q\x\x\q\q\q\q\x\x\q\x\x\q\mynl
\x\x\q\x\x\q\q\q\x\x\x\q\q\q\q\q\q\q\q\q\x\x\q\x\x\q\q\q\q\q\q\q\q\q\x\x\q\q\q\q\x\x\q\x\x\q\mynl
\q\x\x\q\x\x\q\q\q\x\x\x\q\q\q\q\q\q\q\q\q\x\x\q\x\x\q\q\q\q\q\q\q\q\q\x\x\q\q\q\q\x\x\q\x\x\q\mynl
\q\q\q\q\q\q\q\q\q\x\x\q\x\x\q\q\q\q\x\q\q\q\q\q\q\q\q\q\q\x\x\q\x\x\q\q\q\x\q\q\q\q\q\q\q\q\q\x\mynl
\x\q\q\q\x\q\q\q\q\q\q\q\q\x\x\q\q\q\q\x\x\q\q\q\q\q\q\q\x\q\x\x\q\q\q\q\q\q\x\x\q\q\q\q\x\x\q\x\x\mynl
\x\q\x\x\q\q\q\q\q\q\q\x\x\q\x\x\q\q\q\q\q\q\q\q\q\q\q\q\q\q\q\x\x\q\x\q\q\q\q\q\q\q\x\x\q\q\q\q\x\x\mynl
\q\q\x\q\q\q\q\x\x\q\q\x\q\q\q\q\q\q\q\q\q\q\x\q\x\x\q\q\q\q\q\q\q\q\q\q\q\q\x\x\q\x\x\q\q\q\q\q\x\q\x\mynl
\q\x\q\q\q\q\x\x\q\q\x\q\q\q\q\x\x\q\x\x\q\q\q\x\x\q\q\q\q\q\q\q\q\q\q\x\q\x\x\q\q\q\q\q\q\q\q\q\q\x\x\x\mynl
\q\q\q\x\q\q\q\q\q\x\q\x\q\q\q\q\q\x\x\q\x\q\q\q\q\x\x\q\q\q\q\q\q\q\q\q\x\x\q\x\x\q\q\q\q\q\q\q\x\x\q\x\x\mynl
\q\q\q\q\x\q\x\x\q\q\q\q\x\x\q\q\q\q\q\q\q\x\q\q\q\q\x\x\q\q\q\q\q\q\q\x\q\q\q\q\x\q\q\x\x\q\q\q\q\x\x\x\q\x\mynl
\q\x\q\q\q\q\q\q\q\q\q\q\q\x\x\q\x\x\q\q\q\q\q\q\q\q\q\x\x\q\q\q\q\x\x\q\x\q\q\q\q\x\x\q\x\q\q\q\x\q\x\q\x\x\x\mynl
\q\q\q\q\q\x\x\q\x\x\q\q\q\q\x\x\q\q\q\q\q\q\q\q\q\q\x\q\x\x\q\q\q\q\q\q\q\q\q\q\q\q\x\q\q\x\x\q\x\x\q\x\x\q\x\x\mynl
\end{minipage}
\end{center}
\caption{\label{fig:R_4_8_57}A $(4, 8)$-coloring of $K_{57}$
extended from a $(4, 7)$-coloring of $K_{48}$.}
\end{figure}

We have developed a graphical tool to make/modify a Ramsey graph by hand
in its adjacency matrix representation.
Applying the tool to the Ramsey graph shown in Figure \ref{fig:R_4_7_48},
a more symmetric one was obtained for $R(4, 7, 48)$.
Namely, we could flip the color of $e_{0,10},~e_{9,13},$ and $e_{12,16}$
without violating any of the required constraints.
Thus, a Ramsey graph $R(4, 8, 56)$ was obtained based exactly on
Figure \ref{fig:R_4_7_48}, whereas $R(4, 8, 57)$ on the modified one.

\begin{theorem}
  \label{Thm:r_4_8_57}
  $R(4, 8)\ge 58$
\end{theorem}

\begin{proof}
The proof is given by the coloring of $K_{57}$
which can be derived from Figure \ref{fig:R_4_8_57}.
\end{proof}

This coloring improves the lower bound for $R(4, 8)$ from 56 to 58.
The search started with 468 unsettled Ramsey variables + 56 Z-variables
and 3480171 unsettled Ramsey clauses + 936 Z-clauses.
It took $1.87\times 10^{6}$ seconds (about 21 days)
on Mac OS X 10.7.5,
Intel$^{\mbox{\textregistered}}$ CORE{\texttrademark} i7 2GHz,
with 8 GB memory.
The relaxations of Z-clauses were invoked two times,
50 \% reduction each time,
about 25 \% of the initial Z-clauses being kept and satisfied at last.

Related information can be found at the author's web site \cite{Fujita}.


\vfill\newpage

\subsection*{Appendix}

Figure \ref{fig:R_4_8_57_list} is the adjacency list version of
the Ramsey graph $R(4, 8, 57)$.
That is the color one graph of the two-coloring of $K_{57}$, and
all other edges are assigned color two.

\begin{figure}[!ht]
\begin{center}
\begin{minipage}{0.55\textwidth}
\tiny
\begin{verbatim}
0: 2 3 5 6 11 12 22 23 25 26 36 37 38 42 43 45 46 49 50
1: 3 4 6 7 12 13 23 24 26 27 37 38 39 43 44 46 47 52 55
2: 0 4 5 7 8 13 14 24 25 27 28 38 39 40 44 45 47 50 51
3: 0 1 5 6 8 9 14 15 25 26 28 29 39 40 41 45 46 50 53
4: 1 2 6 7 9 10 15 16 26 27 29 30 40 41 42 46 47 49 54
5: 0 2 3 7 8 10 11 16 17 27 28 30 31 41 42 43 47 56
6: 0 1 3 4 8 9 11 12 17 18 28 29 31 32 42 43 44 52 54 56
7: 1 2 4 5 9 10 12 13 18 19 29 30 32 33 43 44 45 51 52 54
8: 2 3 5 6 10 11 13 14 19 20 30 31 33 34 44 45 46 51 56
9: 3 4 6 7 11 12 14 15 20 21 31 32 34 35 45 46 47 48 53 56
10: 4 5 7 8 12 13 15 16 21 22 32 33 35 36 46 47 48 52
11: 0 5 6 8 9 13 14 16 17 22 23 33 34 36 37 47 50 51 53
12: 0 1 6 7 9 10 14 15 17 18 23 24 34 35 37 38 48 50 54
13: 1 2 7 8 10 11 15 16 18 19 23 24 25 35 36 38 39 48 49 54 55
14: 2 3 8 9 11 12 16 17 19 20 24 25 26 36 37 39 40 49 50 55 56
15: 3 4 9 10 12 13 17 18 20 21 25 26 27 37 38 40 41 50 52 56
16: 4 5 10 11 13 14 18 19 21 22 26 27 28 38 39 41 42 52 55
17: 5 6 11 12 14 15 19 20 22 23 27 28 29 39 40 42 43 53 55
18: 6 7 12 13 15 16 20 21 23 24 28 29 30 40 41 43 44 48 52 53
19: 7 8 13 14 16 17 21 22 24 25 29 30 31 41 42 44 45 49 52
20: 8 9 14 15 17 18 22 23 25 26 30 31 32 42 43 45 46 49 53
21: 9 10 15 16 18 19 23 24 26 27 31 32 33 43 44 46 47 54
22: 0 10 11 16 17 19 20 24 25 27 28 32 33 34 44 45 47 51
23: 0 1 11 12 13 17 18 20 21 25 26 28 29 34 35 45 46 52
24: 1 2 12 13 14 18 19 21 22 26 27 29 30 35 36 46 47 51 52
25: 0 2 3 13 14 15 19 20 22 23 27 28 30 31 36 37 47 51 53
26: 0 1 3 4 14 15 16 20 21 23 24 28 29 31 32 37 38 53 54 56
27: 1 2 4 5 15 16 17 21 22 24 25 29 30 32 33 38 39 54 55
28: 2 3 5 6 16 17 18 22 23 25 26 30 31 33 34 39 40 49 55 56
29: 3 4 6 7 17 18 19 23 24 26 27 31 32 34 35 40 41 48 56
30: 4 5 7 8 18 19 20 24 25 27 28 32 33 35 36 41 42 48 49
31: 5 6 8 9 19 20 21 25 26 28 29 33 34 36 37 42 43 49 50
32: 6 7 9 10 20 21 22 26 27 29 30 34 35 37 38 43 44 48 50
33: 7 8 10 11 21 22 27 28 30 31 35 36 38 39 44 45 48 55
34: 8 9 11 12 22 23 28 29 31 32 36 37 39 40 45 46 50 55
35: 9 10 12 13 23 24 29 30 32 33 37 38 40 41 46 47 52 54
36: 0 10 11 13 14 24 25 30 31 33 34 38 39 41 42 47 53 55
37: 0 1 11 12 14 15 25 26 31 32 34 35 39 40 42 43 48 52 53
38: 0 1 2 12 13 15 16 26 27 32 33 35 36 40 41 43 44 49 51 52
39: 1 2 3 13 14 16 17 27 28 33 34 36 37 41 42 44 45 49 51 53
40: 2 3 4 14 15 17 18 28 29 34 35 37 38 42 43 45 46 53 54
41: 3 4 5 15 16 18 19 29 30 35 36 38 39 43 44 46 47 51 55
42: 0 4 5 6 16 17 19 20 30 31 36 37 39 40 44 45 47 50 51 55 56
43: 0 1 5 6 7 17 18 20 21 31 32 37 38 40 41 45 46 50 54
44: 1 2 6 7 8 18 19 21 22 32 33 38 39 41 42 46 47 49 54 55
45: 0 2 3 7 8 9 19 20 22 23 33 34 39 40 42 43 47 49 56
46: 0 1 3 4 8 9 10 20 21 23 24 34 35 40 41 43 44 56
47: 1 2 4 5 9 10 11 21 22 24 25 35 36 41 42 44 45 48 49
48: 9 10 12 13 18 29 30 32 33 37 47 49 50 51 53 55 56
49: 0 4 13 14 19 20 28 30 31 38 39 44 45 47 48 50 52 53 54 56
50: 0 2 3 11 12 14 15 31 32 34 42 43 48 49 51 52 54 55
51: 2 7 8 11 22 24 25 38 39 41 42 48 50 52 53 54 56
52: 1 6 7 10 15 16 18 19 23 24 35 37 38 49 50 51 53 55 56
53: 3 9 11 17 18 20 25 26 36 37 39 40 48 49 51 52 54 55
54: 4 6 7 12 13 21 26 27 35 40 43 44 49 50 51 53 55 56
55: 1 13 14 16 17 27 28 33 34 36 41 42 44 48 50 52 53 54 56
56: 5 6 8 9 14 15 26 28 29 42 45 46 48 49 51 52 54 55
\end{verbatim}
\end{minipage}
\end{center}
\caption{\label{fig:R_4_8_57_list}A $(4, 8)$-coloring of $K_{57}$.}
\end{figure}
\end{document}